\documentclass[11pt]{article}

\usepackage{amsmath,amssymb,amstext,amsthm,url,mathptmx, url}
\usepackage[usenames]{color}

\renewcommand{\ge}{\geqslant}
\renewcommand{\le}{\leqslant}
\newcommand{\poly}{\text{poly}}
\newcommand{\trdeg}{\text{trdeg}}
\newcommand{\spa}{\text{span}}

\usepackage{amsthm}
\theoremstyle{plain}
\newtheorem{theorem}{Theorem}
\newtheorem{conjecture}{Conjecture}

\newtheorem{definition}{Definition}
\newtheorem{lemma}{Lemma}

\theoremstyle{remark}
\newtheorem*{remark}{Remark}

\begin{document}


\title{PIT for depth-$4$ circuits and Sylvester-Gallai conjecture for polynomials}

\author{Alexey Milovanov \\
almas239@gmail.com}

\maketitle

\begin{abstract}
This text is a development of  preprint \cite{gupta}. 

We present an approach for devising a deterministic polynomial time blackbox identity testing (PIT) algorithm for depth-$4$ circuits with bounded top fanin.
This approach is similar to Kayal-Shubhangi \cite{ks} approach for depth-$3$ circuits. Kayal and Shubhangi  based their algorithm on Sylvester-Gallai-type theorem about linear polynomials. We show how it is possible to generalize this approach to depth-$4$ circuits. However we failed to implement this plan completely. We succeeded to construct a polynomial time deterministic algorithm for depth-$4$ circuits with bounded top fanin and its correctness requires a hypothesis. Also we present a  polynomial-time (unconditional) algorithm for some subclass of   depth-$4$ circuits with bounded top fanin.
\end{abstract}

\section{Introduction}

\textbf{Polynomial Identity Testing} : 
In blackbox polynomial identity testing (PIT), given only
query access to a hidden circuit, one has to determine if it outputs the zero polynomial.
In whitebox PIT one has to solve the same problem with possibility to see a circuit.

This problem has numerous applications and has appeared in many fundamental results in
complexity theory. Although this problem exhibits a trivial randomized algorithm, designing
an efficient deterministic algorithm is one of the most challenging open problems. Strong
equivalence results between derandomizing PIT and proving super-polynomial circuit lower
bounds for explicit polynomials are known (cf. Chapter 4 of \cite{sy}).

\textbf{Depth-$4$ Circuits} : In a surprising result, Agrawal-Vinay \cite{av} showed that a complete
derandomization of PIT for just depth-4 ($\Sigma \Pi \Sigma \Pi$) circuits implies an exponential lower bound
for general circuits and a near complete derandomization of PIT for general circuits of poly-degree.  Hence the problem of derandomizing PIT for such fanin
restricted depth-$4$ circuits is equivalent to the general case. 

There has been an incredibly large number of results for $\Sigma \Pi \Sigma \Pi$-circuits with diverse restrictions.
 A study for the case in which the bottom fan-in of such
depth-$4$ circuits is at most $1$ ( $\Sigma \Pi \Sigma$ circuits) was initiated by Dvir-Shpilka \cite{ds}
(whitebox) and Karnin-Shpilka \cite{ksh} (blackbox). A different study for the case with the
restriction of bounded transcendence degree was initiated by Beecken et al. \cite{bms}. Recently, Agrawal
et al. \cite{asss} reproved all these diverse results using a single unified technique based on
the Jacobian criterion. In allmost all these results, the fanin of the top $+$ gate is assumed
to be $O(1)$. For details see the survey by Shpilka-Yehudayoff \cite{sy} or the one by Saxena\cite{s}.

\textbf{The Model} : In this work we consider the model of $\Sigma \Pi \Sigma \Pi(k, r)$ circuits over $\mathbb{C}$, the field
of complex numbers. We first define $\Sigma \Pi \Sigma \Pi(k)$ circuits. These are circuits having four
alternating layers of  and   gates where the fanin of the top  gate is $k$. Such a circuit
alternating layers of $+$ and $\times$ gates where the fanin of the top $+$ gate is $ \le k$. Such a circuit
$C$ computes a polynomial of the form
\begin{equation}
\label{eq1}
C(x_1 , \ldots  , x_n ) = \sum_{i=1}^k F_i =\sum_{i=1}^k \prod_{j=1}^{d_i} l_{ij}
\end{equation}
where $d _i $ are the fanins of the $\times$ gates at the second level. Define $\text{gcd}(C) := \text{gcd}(F_1 , \ldots , F_k )$.   A circuit
is called \emph{simple} if $\text{gcd}(C) = 1$. 
The polynomial computed by a $\Sigma \Pi \Sigma \Pi(k, r)$  circuit $C$ has the same form as in \eqref{eq1}
with  added restriction that the degree of every $l_{ij}$ is at most $r$. As $l_{ij}$  can have at most $r$
irreducible factors, we can factor $l_{ij}$ while incurring a multiplicative factor of $r$ in  $d_i$.
Hence, the polynomial computed by a $\Sigma \Pi \Sigma \Pi(k, r)$ circuit $C$ is of the form
\begin{equation}
\label{eq2}
C(x_1 , \ldots  , x_n ) = \text{gcd}(C) \cdot \sum_{i=1}^k F_i =\sum_{i=1}^k \prod_{j=1}^{d'_i} l'_{ij}
\end{equation} 
where gcd($C$) is a product of polynomials of degree at most  $r$ and $l'_{ij}$ are irreducible. Such a circuit is said to be homogenous if all $F_i$  are homogenous  of
the same degree (and therefore $l'_{ij}$ are also homogenous).

\section{Results}
To state the result we first  need to introduce some notions  from incidence geometry.
\subsection*{Sylvester-Gallai type problems}
  A well-known theorem in incidence geometry called
the Sylvester-Gallai (SG) theorem states that : if there are $n$ distinct points on the real
plane such that, for every pair of distinct points, the line through them also contains a third point,
then they all lie on the same line. Over several decades, various variants of this result have
been proved and are in general called Sylvester-Gallai type problems. Informally, in such
problems, one is presented with a set of objects (points, hyperplanes, etc.) with a lot of
``local'' dependencies (e.g. two points are collinear with a third) and the goal is to translate
these local restrictions to a global bound (usually on the dimension of the space spanned by
the objects). Recently, in an impressive work by Barak et al.\cite{bdwy}, a robust variant
of the SG theorem was proved which among other things says that, even if for every point, the
above stated restriction holds for a constant fraction of other points, one can still bound the
dimension of the vector space spanned by the point set in $\mathbb{C}^d$ by a constant. Few other lines
of study for the SG type problems include
\begin{itemize}
\item replacing lines by higher dimensional vector spaces (initiated by Hansen),
\item having multiple sets of (colored) points (initiated by Motzkin-Rabin),
\item robust/fractional versions of the above (initiated by Barak et al.).
\end{itemize}
For an introduction to the SG theorem and its variants see the survey by Borwein-Moser \cite{bm}.
One interesting feature of \cite{bdwy} is that the robust variant of SG theorem was motivated
 by a problem in theoretical computer science, in particular the study of (linear) Locally
Correctable Codes. A common feature of all these variants is that they only consider flats/vector spaces/linear varieties.
Ankit Gupta and the author propose a new line of SG theorems for non-linear polynomials. 
These problems arise very naturally in our approach for devising PIT algorithms for $\Sigma \Pi \Sigma \Pi(k, r)$ circuits.

The SG theorem can be restated in terms of polynomials as follows: let $l_1 , \ldots , l_m$ be distinct homogenous linear polynomials in 
$\mathbb{R}[x_0 , \ldots , x_n ]$ s.t. for every pair of distinct $l_i$ , $l_j$ there is a
distinct $l_k$ s.t.  $ l_k $ belongs to the ideal $\langle l_i, l_j \rangle$. Then dimension of the vector space spanned by all $l_n$ is
at most $1$.

The dimension of the vector space spanned by a set of linear polynomials is a special case of the general concept of transcendence degree of a set of polynomials .
Polynomials $f_1, \ldots , f_m \subset  \mathbb{C}[x_1,  \ldots , x_n ]$ are called algebraically independent if there is no
non-zero polynomial
$F$
such that $F(f_1, \ldots, f_m)=0$.
The transcendence degree   $\trdeg_{\mathbb{C}}  \{f_1, \ldots , f_m \}$ is the maximal number
$r$ of algebraically independent polynomials in the set.

\begin{definition} 
 A simple  homogenous $\Sigma \Pi \Sigma \Pi(k)$ circuit $C$ such that
\begin{equation}
\label{sg}
C := \sum_{i=1}^{i=k} F_i = \sum_{i=1}^k \prod_{j=1}^{d_i} l_{ij}
\end{equation}
  as stated in Equation (\ref{eq1}) is \emph{SG} if for every $i \in \{1,\ldots, k\}$ and 

for every
$l_{1 j_1}, l_{2 j_2} \ldots, l_{ i-1, j_{i-1}}, l_{ i+1, j_{i+1}}, \ldots,   l_{ k, j_{k}}$ 

the ideal $\langle l_{1 j_1}, l_{2 j_2} \ldots, l_{ i-1, j_{i-1}}, l_{ i+1, j_{i+1}}, \ldots,   l_{ k, j_{k}} \rangle$ contains $F_i$.
\end{definition}
 
Our motivation behind terming such circuits as SG comes from Dvir-Shpilka’s idea of using
variants of the SG theorem for bounding the dimension of the vector space spanned by the
linear forms occurring (at the third layer) in such circuits in the case  the bottom
fanin is at most $1$, i.e., it is a $\Sigma \Pi \Sigma \Pi(k,1)$ circuit. They also conjectured that, if $\mathbb{F}$ has
characteristic $0$ then, this dimension is bounded by a function of only $k$. Indeed later,
Kayal-Saraf \cite{ks} used a colored higher-dimensional variant of the SG theorem to prove this
conjecture for $\mathbb{R}$. In spirit of Dvir-Shpilka \cite{ds} we conjecture that in such SG-$\Sigma \Pi \Sigma \Pi(k,r)$
circuits the transcendence degree of the set of $l_{ij}$ is bounded by a function of $k$, $r$.

\begin{conjecture}
\label{con}
Let $C$ be a  $\Sigma \Pi \Sigma \Pi(k,r)$ circuit of the form (\ref{sg}). If $C$ is SG then $\trdeg_{\mathbb{C}} \{l_{ij}\}  \le \lambda(k, r )$ for some function $\lambda$.

\end{conjecture}
For the case $r = 1$ this conjecture reduces to the case $c = 1$ by the irreducibility of vector
spaces and was first proved over $\mathbb{R}$ in \cite{ks}.
We are now ready to state our first result for PIT.
\begin{theorem} 
\label{Main}
Given  white-box access to  a $\Sigma \Pi \Sigma \Pi(k,r)\text{-circuit}$ $f \in \mathbb{C}[x_0, \ldots, x_n]$ of degree $d$, the identity test for $f$ can be decided deterministically in time $\text{poly}(n, d)$ for constant $k$ and $r$ if $C$ is not SG. Moreover, if Conjecture~\ref{con} holds then the same is true even if $C$ is SG.
\end{theorem}
\begin{remark}
In \cite{gupta} Ankit Gupta gives  another definition of SG-circuit (he uses radical ideals instead of usual ideals). Our approach is in some sense better: we obtain a similar result as in \cite{gupta} under a weaker conjecture.  
\end{remark}
Our second result is a proof of Conjecture~\ref{con} in a special case---Theorem ~\ref{th_pr_con} in Section~\ref{proof_con}. Also we obtain a full deradomization of PIT for some subclass of $\Sigma \Pi \Sigma \Pi(3, 2)$---Theorem~\ref{derand} in Section~\ref{der}.

\section{Case $\Sigma \Pi \Sigma \Pi(3,1)$-circuits}
\label{naive}
Here we present our idea of derandomization for  $\Sigma \Pi \Sigma \Pi(3,1)$-circuits, i.e., polynomials of the form $F_1 + F_2 +F_3$, where every $F_i$ is a product of linear homogeneous polynomials $l_{i1},  l_{i2} \ldots $ 

We can assume w.l.og. that $F_1$, $F_2$ and $F_3$ are pairwise coprime. Indeed, if, say, $(F_1, F_2) \not=1$ then  either $(F_1, F_2) | F_3 $ and we can devide all $F_i$ by $(F_1, F_2)$, or if $F_3$ does not divide $(F_1, F_2)$ then  $F_1 + F_2 + F_3$ is not identiacally zero. If $F_1 + F_2 + F_3  = 0$ then $F_3 \in \langle F_1, F_2 \rangle$ and hence $F_3 \in \langle l_{11}, l_{21} \rangle$. Can we verify the last belonging effectively? The answer is ``yes''.  First, note that  $ \langle l_{11}, l_{21} \rangle$ contains $F_3=l_{31} \cdot l_{32} \ldots$ iff there exists $i$ such that $l_{3i} \in \langle l_{11}, l_{21} \rangle$ because the ideal $ \langle l_{11}, l_{21} \rangle$ is prime.  Now note that we can easily verified whether    $l_{3i} \in \langle l_{11}, l_{21} \rangle$     for every $i$ in polynomial time.

So, we can verify that $\langle l_{1j}, l_{2j} \rangle$ contains $F_3$ for every $i$ and $j$. Similar  for $F_1$ and $F_2$. Assume that we have verified all this and does not find contradictions with zero indentity of $F_1 + F_2 + F_3$. Does this means that this polynomial is zero? No! A conter-example is $F_1:=x$, $F_2:=y$, $F_3:=x + 2y$.

However, the following result shows that in this case the dimension spanned by $l_{ij}$ is at most $4$.  Hence, it is easy to determine the identity of the circuit by Schwartz-Zippel lemma.
\begin{theorem}
		\label{main}
		If $\{S_i\}$ is a finite collection of two or more non-empty disjoint
		finite sets in an affine or in a projective
		complex space such that $\bigcup S_i$
		spans a subspace of at least dimension $5$, then there exists a line 
		 cutting precisely two of the sets. 
	\end{theorem}
	In fact the proof of Theorem~\ref{main} is closely follows the proof  of Edelstein-Kelly theorem in \cite{ek}. We just use the following  result of Kelly instead of Sylvester-Gallai theorem.
	\begin{theorem}[\cite{kelly}]
		\label{kel}
		If a finite set of $k > 2$ points in an affine or in a projective
		complex space  is not a subset of a plane, then there exists a line in that space containing
		precisely two of the points.
	\end{theorem}
	\begin{proof}[Proof of Theorem~\ref{main}]
		First note that a pencil of lines in an affine or a projective $4$-space, not all in the same $3$-dimensional plane must contain a pair of lines such that the plane defined by these lines contains none of the other lines.
		This follows at once if we consider a section of the pencil by a $3$-dimensional plane and
		appeal to  Theorem~\ref{kel} in the $3$-dimensional plane of the section. We call this fact Motzkin's observation since he observed it for $\mathbb{R}$ in \cite{motzkin}. 
		
		We now choose a pair of points $p_1$ and $p_2$ of $\bigcup S_i$ where $p_1$ and $p_2$ are from
		different $S_i$.
		 The points of $\bigcup S_i \setminus \{p_1, p_2\}$
		 define a pencil of $2$-dimensional planes with line
		$p_1 p_2$ as axis.  A section of this pencil by a properly chosen $4$-space defines a
		pencil of lines in that $4$-space not all in a plane. (Indeed, since points of $\bigcup S_i$ do not belongs to any $4$-space there exist points $A, B, C, D$ such that the vectors $p_2 A$, $p_2 B$, $p_2 C$, $p_2 D$, $p_2 p_1$ are linearly independent. The $4$-space $p_2A,p_2B, p_2C, p_2D$ is suitable for us.) By the Motzkin's observation, two of the lines of this pencil define a $2$-plane free of any of the other lines of
		the pencil. This plane together with the points $p_1$ and $p_2$ spans a $3$-space $\Gamma$
		such that the points of $\bigcup S_i$
		in this $3$-space are on precisely two $2$-planes of the
		original pencil of $2$-planes. Each of these planes contains at least one point of $\bigcup S_i \setminus \{p_1, p_2\}$.
		
		Now it is easy  to check that if a collection of two or more finite
		non-empty and disjoint sets in a $3$-dimensional space
		lie on two planes and not on one, then there
		is a line intersecting precisely two of the sets. Indeed, denote these planes as $\alpha$ and $\beta$. If there exist two points  from  $\alpha \bigcup \beta \setminus \{p_1, p_2\}$ from different sets then the lines that connect these points is what we want. Else we  consider any line that connects some point from    $\alpha \bigcup \beta \setminus \{p_1, p_2\}$ and $p_1$ or $p_2$.
\end{proof}
\section{General case $\Sigma \Pi \Sigma \Pi(k, r)$-circuits}	
Now we will try to use the same idea for general 		$\Sigma \Pi \Sigma \Pi(k, r)$-circuits. To simplify notation we  consider $\Sigma \Pi \Sigma \Pi(3,2)$-circuits. So, we consider the circuits of the form $F_1 + F_2 +F_3$, where $F_i$ is a product of linear or quadratic (irredicuble) homogenuos polynomials $l_{i1} l_{i2} \ldots $. We can assume that $F_1$, $F_2$ and $F_3$ are pairwise coprime by the same reasons as before.  Again, we want to verify whether $F_3 \in \langle l_{11}, l_{21} \rangle$. However,  it is not as simple as in the previous section. Membership of $F_3 \in \langle l_{11}, l_{21} \rangle$ does not mean that there exists $l_{3i}$ such that $l_{3i} \in  \langle l_{11}, l_{21} \rangle$. We use the following analogue of this statement.
\begin{theorem}[\cite{mo}]
\label{mathoverflow}
Let $P_1,\ldots, P_d, Q_1, \ldots, Q_k \in \mathbb{C}[x_0,\ldots, x_n]$ be homogenous polynomials of degree at most $r$.

Assume that $P_1 \cdot P_2 \cdots P_{d-1} \cdot P_d \in \langle Q_1, \ldots, Q_k \rangle$. 

Then there exist $\{ i_1, \ldots, i_f \} \subseteq \{1, \ldots, d\}$, where  $f= f(k,r)$ such that
the polynomial $P_{i_1} \cdots P_{i_f} \in I:= \langle Q_1, \ldots, Q_k \rangle$.
\end{theorem}
The proof of this theorem was given by Hailong Dao at MathOverflow~\cite{mo}. We present it here for the convenience of the reader.
\begin{proof}
The point is that many invariants of the ideal $I=(Q_1,\dots,Q_k)$ can be bounded depending only on $k$ and $r$:
\begin{theorem}[\hbox{\cite[Proposition 4.6]{ess}}, \cite{ful98}, \cite{eis95}, \cite{hh}] 
\label{math}
\textup{ }
\begin{enumerate}
\item There exists a  primary decomposition of $I = I_1\cap\dots \cap I_l$ such that each of the $I_i$ is $\mathfrak p_i$-primary and the number of generators of $I_i$ as well as degrees and $l$ itself are bounded by some function of $k$ and $r$.

\item If $I$ is $\mathfrak p$-primary then the minimal $B$ such that $\mathfrak p^B\subseteq I$ is upper bounded by some function from $k$ and $r$.
\end{enumerate}
\end{theorem}
By the first item of this theorem the problem reduces to the case when $I$ is $\mathfrak p$-primary. 
By the second item  there is  $B$ such that $\mathfrak p^B\subseteq I$, and this number  is also  bounded by the degrees and number of generators of $I$.  Remove all the $P_i$  that is not in $\mathfrak p$. The product of the rest is still in $I$ because $I$ is a primary ideal. If there are at most $B$ elements remaining, we are done. If not, then choose $B$ of them, the product is in $\mathfrak p^B\subseteq I$.
\end{proof}
By this theorem it is simple to recognize   membership of $F_3$ in $\langle l_{11}, l_{21} \rangle$ by a polynomial time-bounded algorithm  that proves the first part of Theorem~\ref{Main}. The second part of this theorem follows from the following  result. 
\begin{theorem}[\cite{bms}]
\label{trd}
Let $C$ be an $m$-variate circuit. Let $f_1,\ldots, f_m$ be $l$-sparse,
$\delta$-degree, $n$-variate polynomials with trdeg $r$. Suppose we have oracle access to the
$n$-variate $d$-degree circuit 
$C′:=C(f_1, \ldots , f_m).$ 
There is a blackbox $\poly(size(C′) \cdot d l \delta)^r$
time test to check $C′= 0$ over $\mathbb{C}$.
\end{theorem}

\begin{proof}[Proof of Theorem~\ref{Main}]
We claim that there exists an algorithm verifying that a given $\Sigma \Pi \Sigma \Pi(k, r)$-circuit is SG. Indeed, let $C$ be a circuit of the form (\ref{sg}).  We need to verify that 
$F_i$ belongs to  ideal $I:= \langle l_{1 j_1}, l_{2 j_2} \ldots, l_{ i-1, j_{i-1}}, l_{ i+1, j_{i+1}}, \ldots,   l_{ k, j_{k}} \rangle$ for every $i$ and for every  $l_{1 j_1}, l_{2 j_2} \ldots, l_{ i-1, j_{i-1}}, l_{ i+1, j_{i+1}}, \ldots,   l_{ k, j_{k}}$ (note that there are only $\poly(n,d)$ such conditions for constant $k$). By Theorem~\ref{mathoverflow} $F_i=l_{i,1} \cdot \ldots \cdot l_{i,d_i}$ belongs to $I$ iff
there exists   $\{ i_1, \ldots, i_f \} \subseteq \{1, \ldots, d\}$ such that $l_{i_1} \cdots l_{i_f} \in I$ where $f= f(k,r)$. So there are only $\poly(n,d)$ such conditions  for constant $k$ and $d$. One such condition can be verifiyed in polynomial time (here it is crucial that all these polynomials are homogeneous). Indeed, a homogeneous polynomial $A$ of degree $a$ belongs to $ \langle B_1, \ldots, B_s  \rangle$, where $B_i$ are homogeneous polynomial of degree $b$ if and only if there exists homogeneous polynomials $D_1, \ldots, D_s$ of degree $a-b$ such that  $A = B_1 \cdot D_1 + \ldots B_s \cdot D_s$. Hence, to verify that  $l_{i_1} \cdots l_{i_f} \in I$ we need to solve a system of $\poly(n,d)$ linear equations.   

Now assume we are given a SG $\Sigma \Pi \Sigma \Pi(k, r)$-circuit  (if a circuit is not SG then it is not identically zero). If Conjecture~\ref{con} holds then the trdeg of this circuit is constant. Then by Theorem~\ref{trd} there exists a polynomial (in $n$ and $d$) algorithm  solving PIT for this circuit. 
\end{proof}
\section{Proof of Conjecture~\ref{con} in a special case}
\label{proof_con}
We do not know the correctness of Conjecture~\ref{con} even for  $\Sigma \Pi \Sigma \Pi(3, 2)$ circuits. For this reason we consider a simple subclass of such circuits.  Namely, we consider circuits with the following property:  all ideals $\langle l_{ik},  l_{jt} \rangle$  for different $i$ and $j$ and for quadratic $ l_{ik}$,  $l_{jt}$ are \emph{prime}. Also, we  need that not all quadratic  polynomials $l_{ij}$ have the same index $i$. 
\begin{theorem}
\label{th_pr_con}
Conjecture~\ref{con} holds for such circuits.
\end{theorem}
\begin{proof}
 Denote the set of quadratic polynomials $l_{ij}$ as $Q$. First, we prove that even the dimension of $\spa (Q)$  is bounded by a constant. 
Indeed, a quadratic polynomial $l$ belongs to $\langle l_1, l_2 \rangle$ where $l_1$, $l_2 \in Q$ iff $l$ is a linear combination of $l_1$ and $l_2$. Besides, every such ideal $\langle l_{ij}, l_{ks} \rangle$ where $l_{ij}, l_{ks} \in Q$ and $i\not=k$ \emph{must} contain a quadratic polynomial $l_{tu}$ where $t\not=i,k$ since the ideal $\langle l_{ij}, l_{ks} \rangle$ is prime. 
Hence $\dim (\spa(Q)) $  is at most $2$ by Theorem~\ref{main}. Here, it is important that $\langle l_1, l_2 \rangle$ is prime  and not all quadratic  polynomials $l_{ij}$ have the same index $i$.

Consider an ideal of the form $\langle l_1, l_2 \rangle$ where $l_1$ and $l_2$ are linear.  Recall, that this ideal is prime. 
Denote by $L$ the set of all $l_{ij}$ such that there exists $l_{kt}$ with $k\not=i$ such that the ideal $\langle l_{ij}, l_{kt} \rangle$ contains some quadratic $l_{fu}$ with $f\not=i, k$. 

\begin{lemma}
\label{dimL}
The dimension of $\spa(L)$ is at most $6$. 
\end{lemma}
\begin{proof}[Proof of Lemma~\ref{dimL}]\
\begin{enumerate}
\item  A quadratic homogeneous polynomial $f$ over $\mathbb{C}$ is irreducible iff $\text{rk}(f) \ge 3$. Here, $\text{rk}(f)$ is the  rank of $f$ as a quadratic form. Indeed, if  $\text{rk}(f) < 3$ then it is obvious that $f$ is not irreducible. To prove that in other cases $f$ is irreducible it is enough to show  that the  polynomial $x^2 + y^2 + z^2$ is irreducible and this is folklor. So, all elements of $Q$ have rank at least $3$. 

\item Denote by $Q'$ the subset of all elements of $q \in Q$ such that there exist $i$, $j$, $k$ and $t$  with $k\not=i$ s. t. $\langle l_{ij}, l_{kt} \rangle \in q$. Of course the dimension of $\spa(Q')$ is at most $2$ as the dimension of $\spa(Q)$. 

\item Consider some $l$, $m \in L$ and $q \in Q'$ such that $\langle l, m \rangle$ contains $q$. This means that the intersection of quadric $Q'$ with line $l$  is a quadric  with rank at most $2$. Therefore, $\text{rk}(q) \le 3$.  Combining this result with the first item we conclude that $\text{rk}(q) = 3$ for every $q \in Q'$.

\item Consider the largest linear independent subset in $L$.  Denote this set as $\{l_1, \ldots, l_t\}$ . We will show that $t \le 3 \cdot \text{dim}(\spa(Q')) \le 6$. This give us what we want.

\item Add new $l'_{t+1},\ldots, l'_n$ such that  $\{l_1, \ldots, l_t, l'_{t+1},\ldots, l'_n\}$ is a basis of the linear form from $x_1, \ldots, x_n$. Consider the (symmetric) matrices $A_1, \ldots, A_s$ of all  quadratics from $Q'$ in the dual basis of $\{l_1, \ldots, l_t, l'_{t+1},\ldots, l'_n\}$.

\item The rank of every $A_i$ is equal to $3$. Hence, there exist $3\cdot \text{dim}(\spa(Q'))$ numbers of rows such that  other rows are linearly depend from  these in \emph{every} matrix $A_i$. The same is true for columns since these matrices are symmetric. 

\item For every $l_j$ the exists $q \in Q'$ such that $l_i \cap q$ is a quadratic form of rank $2$. Hence, for every $i = 1, \ldots, t$ there exists a matrix $A_j$ such that matrix $A_{j,i}$  obtained from $A$ by deleting the $i$th row and the $i$th column has rank $2$. But from 6, it follows that there are at most $3 \cdot \text{dim}(\spa(Q'))$ such numbers $i$. Therefore $t \le 3\cdot \text{dim}(\spa(Q'))$. 
\end{enumerate}
\end{proof}
Add to the set $L$ the polynomials $l_{ij}$ that are linear combinations of $L$. 
 Lemma~\ref{dimL} shows that the dimension of  $L$ is not greater than $6$. 
We need to prove that the dimension of the span of the remaning linear polynomials is also bounded by a constant. Denote the set of such polynomials by $T$. 
The elements of $T$ have the following property. If $l_{ij} \in T$ and $l_{ts} \in T\cup L$ with $i \not=t$ then  there exist $l_{pu} \in T \cup L$ such that 
$p\not= i,t$  and $l_{pu}$ is a linear combination of $l_{ij}$ and $l_{ts}$.
Note that we can not say that if $l_{ij} \in T$ and $l_{ts} \in T$ then there exists $l_{pu} \in T$ that is a linear combination of $l_{ij}$ and $l_{ts}$, so we can not apply Theorem~\ref{main} directly. 
However the idea of the proof of Theorem~\ref{main} works.

We claim that dim$(\spa(T\cup L)) < \text{dim } (\spa(L)) + 4=10$. Together with dim $(\spa(Q)) \le 2$ this implis that $\trdeg$ of all polynomials $l_{ij}$ is less than $12$ (this proves the theorem). 

Assume that  dim$(\spa(T\cup L)) \ge \text{dim } (\spa (L)) + 4$.
Devide $T \cup L$ in three sets $S_1$, $S_2$ and $S_3$ in a natural way (in accordance with indexes $i$ of $l_{ij}$).
As in the proof of Theorem~\ref{main}, take $p_1 \in S_1$ and $p_2 \in S_2$.
Again we consider the pencil of $2$-dimensional planes with line $p_1 p_2$ as axis.
Since  dim$(\spa(T\cup L)) < \text{dim } (\spa (L)) + 4$, there exist $t_1, t_2, t_3, t_4 \in T$ such that $p_2t_1$, $p_2t_2$, $p_2t_3$ and $p_2t_4$ are linear independent and there are no points from $L$ in subspace $p_2 t_1 t_2  t_3 t_4$.    As in the proof of Theorem~\ref{main} we can conclude that there exist points $T_1$, $T_2 \in T$ such that  in the 3-space plane generated by $p_1$, $p_2$, $T_1$ and $T_2$ all points from $T\cup L$ belong to two $2$-spaces $p_1 p_2 T_1$ and $p_1 p_2 T_2$.

If $T_1$ and $T_2$ are from different $S_i$ then we get a contradict (there are no another points from $L \cup T$ at line $T_1T_2$).
 Otherwise, all points in 3-space $p_1 p_2 T_1 T_2 \setminus \{p,q\} $ belong to one $S_i$. Then we get a contradiction  considering line $T_1 p_1$ or $T_1 p_2$. 
\end{proof}

\section{Derandomization of PIT for some subclass of  $\Sigma \Pi \Sigma \Pi(3, 2)$ circuits}
\label{der}
In Theorem~\ref{th_pr_con} we have the strange condition that not all quadratic  polynomials $l_{ij}$ have the same index $i$. 
To cover this case we present an algorithm solving PIT for such circuits. 

More precisely, we consider $\Sigma \Pi \Sigma \Pi(3, 2)$-circuits of the form $F_1 + F_2 + F_3$,
where $F_1$ and $F_2$ are products of homogenous linear polynomials and
$F_3$ is a product of homogeneous quadratic and linear polynomials. 
\begin{theorem}
\label{derand}
There exists a polynomial-time algorithm solving PIT for such circuits. 
\end{theorem}
\begin{proof}
Let $F_i = l_{i1} \cdot \ldots  $ for $i=1,2$ (here $l_{ij}$ are linear polynomials ) and $F_3 = q_{31} \cdot \ldots q_{3s} \cdot l_{31} \cdot \ldots l_{3r}$ 
(here $q_{3j}$ are irredicable quadratic and $l_{3j}$  are linear polynomials). We assume that $s \ge 1$(otherwise we can just use results of Section~\ref{naive}).

We can assume that $l_{11} = x$.  Then $F_1 + F_2 + F_3 = 0$ implies $F_{2}|_{x=0} + F_{3}|_{x=0} =0$. Since $F_{2}|_{x=0}$ is a product of linear polynomials,  $q_{31}|_{x=0}$ must be factorized. 
Hence (see the proof of Lemma~\ref{dimL}), the rank of $q_{31} = 3$. The polynomial $q_{31}|_{l_{ij}=0 }$ must be factorized for all linear polynomials from $F_1$ and $F_2$ (otherwise $F_1 + F_2 + F_3 \not= 0$). This implies that  the dimension of linear forms spanned by polynomials from $F_1$ and $F_2$ is at most $3$.

In other words $F_1$ and $F_2$ depends only on $3$ variables (after linear changing or variables). If $F_3$ depends  another variables, then a given circuit is not identically zero. Otherwise, 
 $\trdeg$ of all $l_{ij}$ and $q_{kt}$ is not greater than $3$. Hence, by Theorem~\ref{trd}  there exists a polynomial-time algorithm for such circuits.

\end{proof}

\section*{Acknowledgments}
I would like to thank Hailong Dao for useful discussions and Bruno Bauwens for help in writing this paper.


\begin{thebibliography}{M}
\bibitem{asss} 
Manindra Agrawal, Chandan Saha, Ramprasad Saptharishi, and Nitin Saxena.
Jacobian hits circuits: hitting-sets, lower bounds for depth-$d$ occur-$k$ formulas
\& depth-$3$ transcendence degree-$k$ circuits. In \emph{Proceedings of ACM Symposium
on Theory of Computing} (STOC), pages 599--614, 2012.


\bibitem{av} Manindra Agrawal and V. Vinay. Arithmetic circuits: A chasm at depth four.
\emph{In Proceedings of IEEE Foundations of Computer Science} (FOCS), pages 67--75,
2008.

\bibitem{bdwy}
Boaz Barak, Zeev Dvir, Avi Wigderson, and Amir Yehudayoff. Fractional
sylvester–gallai theorems. \emph{Proceedings of the National Academy of Sciences},
110(48):19213--19219, 2013.

\bibitem{bm} Peter Borwein and William OJ Moser. A survey of sylvester’s problem and its
generalizations. \emph{Aequationes Mathematicae}, 40(1):111--135, 1990.

\bibitem{bms} Malte Beecken, Johannes Mittmann, and Nitin Saxena. Algebraic independence
and blackbox identity testing. \emph{Information and Computation}, 222:2--19, 2013.



\bibitem{ds}
Zeev Dvir and Amir Shpilka. Locally decodable codes with 2 queries and polynomial identity testing for depth 3 circuits. \emph{SIAM Journal on Computing},
36(5):1404--1434, 2006.


\bibitem{ek}
M. Edelstein and L. M. Kelly, 
{\em Bisecants of finite collections of sets in linear space},
Canadanian Journal of Mathematics, 18:375--380, 1966,
\url{https://cms.math.ca/10.4153/CJM-1966-039-2} 

\bibitem{ess}
Daniel Erman, Steven V Sam  and Andrew Snowden,
{\em Generalizations of Stillman's conjecture via twisted commutative algebras }
\url{https://arxiv.org/abs/1804.09807}

\bibitem{eis95} David Eisenbud, {\em Commutative algebra with a view toward algebraic geometry}
, Graduate Texts in
Mathematics, vol. 150, Springer-Verlag, New York, 1995.

\bibitem{ful98} William Fulton, {\em Intersection theory}, 2nd ed., Ergebnisse der Mathematik und ihrer Grenzgebiete.
3. Folge. A Series of Modern Surveys in Mathematics [Results in Mathe
matics and Related Areas.
3rd Series. A Series of Modern Surveys in Mathematics], vol. 2, Sprin
ger-Verlag, Berlin, 1998.

\bibitem{gupta}
Ankit Gupta, Algebraic Geometric Techniques for

Depth-4 PIT and Sylvester-Gallai Conjectures for Varieties. 

\url{https://eccc.weizmann.ac.il/report/2014/130/}

\bibitem{mo}
Hailong Dao, 
\url{https://mathoverflow.net/questions/288630/generators-of-an-ideal-with-small-degree}


\bibitem{hh} Melvin Hochster and Craig Huneke,
{\em Comparison of symbolic and ordinary powers of ideals} , Invent. Math. 147 (2002), no. 2, 349–369.


\bibitem{ksh}
Zohar Shay Karnin and Amir Shpilka. Black box polynomial identity testing of
generalized depth-3 arithmetic circuits with bounded top fan-in. \emph{Combinatorica},
31(3):333–364, 2011.

\bibitem{ks}
Neeraj Kayal and Shubhangi Saraf. Blackbox polynomial identity testing for
depth 3 circuits. In \emph{Proceedings of IEEE Foundations of Computer Science}
(FOCS), 2009.

\bibitem{kelly}
L. M. Kelly, A resolution of the Sylvester-Gallai problem of J.-P. Serre,
\emph{Discrete Comput. Geom.} 1 (1986), 101--104.
		

\bibitem{motzkin}
		 Th. Motzkin, The lines and planes connecting the points of a finite set,
		 \emph{ Trans. Amer. Math. Soc}, 70 (1951), 451--464. 

\bibitem{s}
Nitin Saxena. Progress on polynomial identity testing --- II.
\url{https://arxiv.org/abs/1401.0976}

\bibitem{sy} Amir Shpilka and Amir Yehudayoff. Arithmetic circuits: A survey of recent
results and open questions.  \emph{Foundations and Trends in Theoretical Computer
Science}, 5(3-4):207–388, 2010.




\end{thebibliography}
\end{document}